\documentclass{article}
\usepackage{times}
\usepackage{epsfig}
\usepackage{graphicx}
\usepackage{amsmath}
\usepackage{amssymb}
\usepackage{comment}

\usepackage[pagebackref=true,breaklinks=true,letterpaper=true,colorlinks,bookmarks=false]{hyperref}

\usepackage{url}

\usepackage{amsmath, amsfonts, amsthm}
\usepackage{float}
\usepackage{mathtools}
\usepackage{xcolor}
\usepackage[numbers]{natbib}

\DeclarePairedDelimiterX\set[1]\lbrace\rbrace{#1}
\newcommand{\dd}{\mathrm{d}}

\newtheorem{theorem}{Theorem}[section]
\newtheorem{proposition}[theorem]{Proposition}
\newtheorem{lemma}[theorem]{Lemma}

\newtheorem{corollary}[theorem]{Corollary}
\theoremstyle{definition}
\newtheorem{ex}[theorem]{Example}
\newtheorem{definition}[theorem]{Definition}

\newtheorem*{conjecture}{Conjecture}

\setcitestyle{square} 

\begin{document}

\title{Gaussian Persistence Curves}

\author{Yu-Min Chung, Michael Hull, Austin Lawson, and Neil Pritchard}

\maketitle

\begin{abstract}

Topological data analysis (TDA) is a rising field in the intersection of mathematics, statistics, and computer science/data science. The cornerstone of TDA is persistent homology, which produces a summary of topological information called a persistence diagram. To utilize machine and deep learning methods on persistence diagrams, These diagrams are further summarized by transforming them into functions. In this paper we investigate the stability and injectivity of a class of smooth, one-dimensional functional summaries called Gaussian persistence curves.

\end{abstract}

\section{Introduction}
\label{sec:intro}

One of the main tools of topological data analysis (TDA) is persistent homology, which measures how certain topological features of a data set appear and disappear at different scales. This information can be stored and visualized in a concise format called a persistence diagram.

Functional summaries play an important role in topological data analysis, as they allow one to apply machine and deep learning techniques to analyze topological information contained in persistence diagrams. In \cite{chung2020smooth} a new class of one-dimensional smooth functional summaries was introduced called Gaussian persistence curves (GPC's). These functional summaries were built by combining (a slight variation of) the persistence curve framework from \cite{chunglawson2019} with the persistence surfaces construction from \cite{adams2017persistence}, and they were used to study the texture classification of grey-scale images \cite{chung2020smooth}.

In this paper, we investigate the stability of GPC's and the injectivity of both persistence surfaces and GPC's. Loosely speaking, stability refers to the property that small changes in diagrams correspond to small changes in the resulting summaries and the injectivity of a summary implies that the summary can distinguish between distinct diagrams. We show that unweighted GPC's are stable (Theorem \ref{t:unwstab}) and that, under mild conditions, weighted GPC's are both stable (Corollary \ref{c:wstab2} and Theorems \ref{t:wstab3} \& \ref{t:wstab4}) and almost injective (Theorem \ref{t:ainj}). Furthermore we show that unweighted persistence surfaces are injective (Theorem \ref{injective}).

Other summaries in topological data analysis include persistence landscapes \cite{bubenik2015statistical}, the persistent entropy summary function \cite{persistentEntropyStability}, persistence silhouettes \cite{chazal2013bootstrap}, persistence surfaces and persistence images \cite{adams2017persistence}. We refer to \cite{berry2018functional} for a review of the properties and applications of these summaries. All of these other summaries are known to be stable, but among them only the persistence landscapes are known to be injective. Note that persistence landscapes can be viewed as a sequences of one-dimensional functions and for any $n\geq 1$ injectivity will fail if only the first $n$ terms of the sequences are considered.

The outline of this paper is as follows. In Section \ref{sec:smooth pc} we introduce Gaussian persistence curves and derive some basic properties and useful formulas. In Section \ref{sec:stability} we apply these formulas to prove stability of unweighted GPC's and certain weighted GPC's. Finally, in Section \ref{sec:inj} we investigate the extent to which the functional summaries produced by persistence surfaces and GPC's are injective.

\section{Gaussian Persistence Curves}
\label{sec:smooth pc}

Here we describe the construction of Gaussian persistence curves which can be viewed as a combination of the persistence surfaces introduced in \cite{adams2017persistence} and the persistence curve framework from \cite{chunglawson2019}. We refer to \cite{edelsbrunner2008persistent} for background on persistent homology and persistence diagrams.

The input to our construction is a persistence diagram, by which we mean a finite multi-set $D$ of points in the plane which lie above the main diagonal $y=x$. Let $\mathbf{\Sigma}$ be a symmetric, positive semi-definite $2\times2$ matrix. For a point $\boldsymbol{\mu}\in\mathbb{R}^2$, Let $g_{\boldsymbol\mu,\mathbf{\Sigma}}$ be the probability density function (PDF) of a bivariate normal distribution with mean $\boldsymbol{\mu}$ and covariance matrix $\mathbf{\Sigma}$. That is, \[g_{\boldsymbol\mu,\Sigma}(\mathbf{x}) = \frac{\exp\left(-\frac{1}{2}(\mathbf{x}-\boldsymbol{\mu})^T\Sigma^{-1}(\mathbf{x}-\boldsymbol{\mu})\right)}{2\pi|\boldsymbol\Sigma|^{1/2}}.\]
Let $\kappa\colon \mathbb R^2\to\mathbb R$ be a function with $\kappa(b,b) = 0$ for all $b\in \mathbb R$. We refer to any such $\kappa$ as a weighting function.
\begin{definition}\cite{adams2017persistence}
The persistence surface associated to the diagram $D$ with weight $\kappa$ is the function
\[
\rho_{D,\kappa}(x, y)=\sum_{(b,d)\in D}\kappa(b,d)g_{(b,d),\boldsymbol{\Sigma}}(x, y).
\]
\end{definition}

In \cite{adams2017persistence} the authors choose a grid on the plane, integrate the persistence surface over each box in the grid, and then use these values to produce a vector summary of the original diagram. Instead, we look at this surface from the perspective of the persistence curve framework from \cite{chunglawson2019}. This framework produces a function $G\colon \mathbb R\to\mathbb R$ such that the value of $G(t)$ depends on measuring some property of the diagram inside the \emph{fundamental box} $F_t=\{(x, y)\;|\; x<t, y>t\}$.

\begin{definition}
\label{def:GaussianPC}
Let $\rho_{D, \kappa}$ be a persistence surface. The corresponding \textbf{Gaussian persistence curve} is the function
\[
G_{D, \kappa}(t)=\int_{F_t} \rho_{D, \kappa}(x, y)dxdy.  
\]
\end{definition}
If $\kappa(x, y)=1$ for all $(x, y)$, then we drop $\kappa$ from the notation and denote the corresponding surface and curve by $\rho_D(t)$ and $G_D(t)$ respectively. We refer to this curve as the unweighted Gaussian persistence curve on $D$.

We will always consider $\Sigma$ to be fixed ahead of time and do not include it in the notation for the curve $G_{D, \kappa}(t)$. While the definition makes sense for more general $\Sigma$, in this paper we fix $\Sigma$ to be a multiple of the identity matrix by a scalar $\sigma^2$. This allows us to split $g_{\boldsymbol\mu,\mathbf{\Sigma}}$  as
\[
g_{\boldsymbol\mu,\mathbf{\Sigma}}= \phi(\frac{y-d}{\sigma})\phi(\frac{x-b}{\sigma})
\]
where $\phi$ is the pdf of the standard normal distribution. This assumption allows one to easily preform the integration over the fundamental box and obtain the \textbf{CDF realization} of $G_{D, \kappa}(t)$ as
\[
G_{D, \kappa}(t) = \sum_{(b,d)\in D} \kappa_D(b,d)\Phi(\frac{t-b}{\sigma})\Phi(\frac{d-t}{\sigma}).
\]

Gaussian persistence curves fit into (a slight modification of) the persistence curve framework from \cite{chunglawson2019} in the following way. Let $\mathcal{D}$ be the set of all persistence diagrams,  $\Psi$ be the set of all functions $\psi:\mathcal{D}\times \mathbb{R}^3\to\mathbb{R}$ with $\psi(D;x,x,t) = 0$ for all $(x,x) \in \mathbb{R}^2$ and $D \in \mathcal{D}$. Let $\mathcal{R}$ represent the set of functions on $\mathbb{R}$. Let $\mathcal{T}$ be a set of operators $T(S,f)$ that read in a multi-set $S$ and real-valued function $f$ and returns a scalar. Given $D\in \mathcal{D}$, $\psi\in\Psi$, and $T\in\mathcal T$, the corresponding persistence curve is the function

\[P_{D,\psi, T}(t) := T(F_t,\psi(D;x,y,t)),~t\in\mathbb{R}.\] 

The function $P_{D,\psi, T}(t)$ is called a \textbf{persistence curve} on $D$ with respect to $\psi$ and $T$. In this notation, choosing $\psi(D;x,y,t)=\rho_{D, \kappa}(x, y)$ and $T(f, S)=\int_S f(x, y)dxdy$, we obtain$P_{D,\psi, T}(t)=G_{D, \kappa}(t)$.

We start with a few examples of Gaussian persistence curves, which are smooth versions of persistence curves appearing in \cite{chunglawson2019}. 
\begin{ex}\label{ex:smoothBetti}
When $\kappa_D(b,d) = 1$ for all $(b, d)$, the resulting unweighted Gaussian persistence curve can be viewed as a smooth version of the Betti curve from \cite{chunglawson2019}. For this reason we also refer to the unweighted Gaussian persistence curve $G_D(t)$ as the \emph{Gaussian Betti Curve}. 

\end{ex}

\begin{ex}\label{ex:smoothlifecurve}
Let $\ell_D(x,y) := (y-x)\cdot\chi_{D}(x,y)$ and let $\ell_{sum} = \sum_{(b,d)\in D}\ell(b,d)$. Define $\kappa_D(x,y) = \frac{\ell(x,y)}{\ell_{sum}}$. The corresponding Gaussian persistence curve is a smooth version of the life curve from \cite{chunglawson2019} which we call the \emph{Gaussian Life Curve}.
\end{ex}

\begin{ex}
Let $m_{sum}=\sum_{(b,d)\in D}(b+d)$ and define $\kappa_D(x,y) = \frac{x+y}{m_{sum}}$. The corresponding Gaussian persistence curve is a smooth version of the midlife curve from \cite{chunglawson2019} which we call the \emph{Gaussian midlife Curve}.
\end{ex}

With this set up, generating new Gaussian persistence curves is only a matter of selecting a covariance matrix $\Sigma$, which controls the smoothness of the curve and a function $\kappa$, which is a weighting function. For example, by using weight functions such as entropy function (-$\frac{d-b}{\sum_{(b,d)\in D}d-b}\log\frac{d-b}{\sum_{(b,d)\in D}d-b}$) and multiplicative life function ($\frac{d}{b})$, we can obtain Gaussian versions of the life entropy and multiplicative life persistence curves. In general, we can produce a Gaussian version of any function in the persistence curve framework. 

The next two lemmas will be used to compute the $L^1$--norm of a Gaussian persistence curve. Their proofs are elementary exercises in Calculus. 

\begin{lemma}\label{lem:base1}
Given $b>0 \in \mathbb{R}$,
\begin{equation}
\int_{-\infty}^{\infty} \Phi(\frac{b-t}{\sigma}) \Phi(\frac{t-b}{\sigma}) \dd{t} = \frac{\sigma}{\sqrt{\pi}}.
\end{equation}
\end{lemma}

\begin{proof}
We will first prove the result for $b=0$ and $\sigma =1$.
By integration by parts (by letting $u=\Phi(-t)$ and $\dd{v} = \Phi(t) \dd{t}$), we obtain
\begin{align*}
\int_{-\infty}^{\infty} \Phi(-t) \Phi(t) \dd{t} = \underbrace{[t\Phi(-t) \Phi(t) + \Phi(-t) \phi(t)]_{-\infty}^{\infty}}_{I} + \underbrace{ \int_{-\infty}^{\infty} t \phi(t) \Phi(t) + \phi^2(t) \dd{t}}_{II}.    
\end{align*}

For $I$, by the elementary facts $\Phi(t) \rightarrow 1$ as $t \rightarrow \infty$ and L'Hopital's rule, one can evaluate that $I = 0$.

For the $II$, we consider each integral separately. By \cite{patel1996handbook},

we have that $\int_{\infty}^{\infty} t \Phi(a + bt) \phi(t) \dd{t} = \frac{b}{\sqrt{1+b^2}} \phi(\frac{a}{\sqrt{1+b^2}})$.  Since in our case $a=0$ and $b=1$, 
\begin{equation*}
    \int_{-\infty}^{\infty} t \Phi(t) \phi(t) \dd{t} = \frac{1}{\sqrt{2}}\phi(0) = \frac{1}{2\sqrt{\pi}}.
\end{equation*}
By \cite{patel1996handbook} again, we know that
\begin{equation*}
    \int_{-\infty}^{\infty} \phi^2(t) \dd{t} = \frac{1}{2\sqrt{\pi}}\Phi(\sqrt{2}t) \bigg|^{\infty}_{-\infty} = \frac{1}{2\sqrt{\pi}}.
\end{equation*}
Thus, sum over them to obtain the desired result. To obtain the final result simply apply the substitution, $s = \frac{t-b}{\sigma}$ then $ds= \frac{1}{\sigma}dt$.
\begin{align*}
    \int_{-\infty}^{\infty} \Phi(\frac{b-t}{\sigma}) \Phi(\frac{t-b}{\sigma}) \dd{t} &= \int_{-\infty}^{\infty} \Phi(-s) \Phi(s) \sigma \dd{s} = \frac{\sigma}{\sqrt{\pi}}.
\end{align*}
    
\end{proof}

\begin{lemma}\label{lem:base2}
\begin{align}
    &\int_{-\infty}^{\infty} \Phi(at + b_1) (\Phi(at+b_2)-\Phi(at+b_3) ) \dd{t}\\
    &= \frac{-\sqrt{2}}{a} \left[ \frac{b_1-b_2}{\sqrt{2}}\Phi(\frac{b_1-b_2}{\sqrt{2}}) + \phi(\frac{b_1-b_2}{\sqrt{2}}) - \frac{b_1-b_3}{\sqrt{2}}\Phi(\frac{b_1-b_3}{\sqrt{2}}) - \phi(\frac{b_1-b_3}{\sqrt{2}}) \right]
\end{align}
\end{lemma}
\begin{proof}
Let $s = u + (at + b_3)$.  Then $\dd{s} = \dd{u}$.
\begin{align*}
    &\int_{-\infty}^{\infty} \Phi(at + b_1) (\Phi(at+b_2)-\Phi(at+b_3) ) \dd{t}\\
    &= \int_{-\infty}^{\infty} \Phi(at + b_1) \int_{{at+b_3}}^{at+b_2} \phi(s)~ \dd{s} \dd{t} \\
    &= \int_{-\infty}^{\infty}  \int_{{at+b_3}}^{at+b_2} \Phi(at + b_1) \phi(s) ~ \dd{s} \dd{t} \\
    &=\int_{-\infty}^{\infty}  \int_{0}^{b_2-b_3} \Phi(at + b_1) \phi(at+b_3 + u) ~ \dd{u} \dd{t} \\
    &= \int_{0}^{b_2-b_3}\int_{-\infty}^{\infty}  \Phi(at + b_1) \phi(at+b_3 + u) ~ \dd{t} \dd{u}.
\end{align*}
To evaluate the integral, we recall that $\int_{-\infty}^{\infty} \Phi(x+\epsilon) \phi(x)~ \dd{x} = \Phi(\frac{\epsilon}{\sqrt{2}})$.  Consider another substitution: $x = at + b_3 + u$, so $\dd{x} = a \dd{t}$.
\begin{align*}
    &\int_{-\infty}^{\infty}  \Phi(at + b_1) \phi(at+b_3 + u) ~ \dd{t} \\
    &=\int_{-\infty}^{\infty} \Phi(x-u-b_3+b_1) \phi(x) \frac{1}{a} ~ \dd{x} = \frac{1}{a} \Phi(\frac{b_1-b_3-u}{\sqrt{2}}).
\end{align*}
Let $v=\frac{b_1-b_3-u}{\sqrt{2}}$. Then $\dd{v} = \frac{-1}{\sqrt{2}} \dd{u}$. Finally, we obtain
\begin{align*}
    &\int_{0}^{b_2-b_3}\int_{-\infty}^{\infty}  \Phi(at + b_1) \phi(at+b_3 + u) ~ \dd{t} \dd{u}\\
    &=\int_{0}^{b_2-b_3}\frac{1}{a} \Phi(\frac{b_1-b_3-u}{\sqrt{2}}) ~ \dd{u}.\\
    &= \int_{\frac{b_1-b_3}{\sqrt{2}}}^{\frac{b_1-b_2}{\sqrt{2}}} \frac{-1}{a} \sqrt{2} \Phi(v) \dd{v} = \frac{-\sqrt{2}}{a} \left[ v\Phi(v) + \phi(v) \right]_{\frac{b_1-b_3}{\sqrt{2}}}^{\frac{b_1-b_2}{\sqrt{2}}}\\
    &=\frac{-\sqrt{2}}{a} \left[ \frac{b_1-b_2}{\sqrt{2}}\Phi(\frac{b_1-b_2}{\sqrt{2}}) + \phi(\frac{b_1-b_2}{\sqrt{2}}) - \frac{b_1-b_3}{\sqrt{2}}\Phi(\frac{b_1-b_3}{\sqrt{2}}) - \phi(\frac{b_1-b_3}{\sqrt{2}}) \right].
\end{align*}
\end{proof}

\begin{proposition}\label{1-normGPC}
Let $G_D(t)$ be an unweighted Gaussian persistence curve on a diagram $D$. Then
\begin{equation}
    \| G_D(t) \|_1 = \sum_{(b,d)\in D} \left[ (d-b) \Phi(\frac{d-b}{\sqrt{2}\sigma}) + \sqrt{2}\sigma\phi(\frac{d-b}{\sqrt{2}\sigma}) \right].
\end{equation}
\end{proposition}
\begin{proof}
\begin{align}
    &\| G_D \|_1 = \int_{-\infty}^{\infty} \left | \sum_{(b,d)\in D} \Phi(\frac{t-b}{\sigma}) \Phi(\frac{d-t}{\sigma}) \right | \dd{t}\\
    &= \int_{-\infty}^{\infty}  \sum_{(b,d)\in D} \Phi(\frac{t-b}{\sigma}) \Phi(\frac{d-t}{\sigma}) \dd{t}\\
    &= \sum_{(b,d)\in D} \int_{-\infty}^{\infty} \Phi(\frac{t-b}{\sigma}) \Phi(\frac{d-t}{\sigma}) \dd{t}.
    \end{align}
    By adding and subtracting $\Phi(\frac{t-b}{\sigma})+\Phi^2(\frac{t-b}{\sigma})$ and using the CDF property $\Phi(-t) = 1-\Phi(t)$ we obtain,
    \begin{align}
    &= \sum_{(b,d)\in D} \int_{-\infty}^{\infty} \Phi(\frac{t-b}{\sigma}) \Phi(\frac{b-t}{\sigma}) + \Phi(\frac{t-b}{\sigma}) \left(\Phi(\frac{t-b}{\sigma})-\Phi(\frac{t-d}{\sigma}) \right) \dd{t} \\
    &= \sum_{(b,d)\in D} \frac{\sigma}{\sqrt{\pi}} +  (d-b) \Phi(\frac{d-b}{\sqrt{2}\sigma}) + \sqrt{2}\sigma\phi(\frac{d-b}{\sqrt{2}\sigma}) - \frac{\sigma}{\sqrt{\pi}}\\
    &= \sum_{(b,d)\in D} \left[ (d-b) \Phi(\frac{d-b}{\sqrt{2}\sigma}) + \sqrt{2}\sigma\phi(\frac{d-b}{\sqrt{2}\sigma}) \right].
\end{align}
where (9) follows from Lemmas~\ref{lem:base1} and \ref{lem:base2}.
\end{proof}

Let $L_D=\sum_{\{(b, d)\in D\}} d-b$. We refer to $L_D$ as the {\bf total lifespan} of $D$. We also define $\delta_D=\min_{(b, d)\in D} d-b$, that is $\delta_D$ is the {\bf minimum lifespan} of a point in $D$. We note that by convention, $\min\emptyset = \infty$ and $\frac{1}{\min\emptyset} = 0$.

\begin{corollary}\label{unweighted1bound}
For any persistence diagram $D$,
\[
\| G_D(t) \|_1 \leq \sum_{(b,d)\in D} \left[ (d-b) + \frac{\sigma}{\sqrt{\pi}} \right]\leq (1+\frac{\sigma}{\sqrt{\pi}\delta_D})L_D.  
\]
\end{corollary}

Apply the same argument as above gives a bound for the weighted case as well. Let $M_{D, \kappa}=\max_{(b, d)\in D}|\kappa(b, d)|$.

\begin{corollary}\label{weighted1bound}
\begin{align*}
\| G_{D, \kappa}(t) \|_1 &= \sum_{(b,d)\in D} |\kappa(b, d)|\left[ (d-b) \Phi(\frac{d-b}{\sqrt{2}\sigma}) + 
\sqrt{2}\sigma\phi(\frac{d-b}{\sqrt{2}\sigma}) \right].\\
&\leq \sum_{(b,d)\in D} |\kappa(b, d)|\left[ (d-b) + \frac{\sigma}{\sqrt{\pi}} \right]\leq (1+\frac{\sigma}{\sqrt{\pi}\delta_D})M_{D,\kappa}L_D
\end{align*}
\end{corollary}


As long as the diagram $D$ is finite, then the Gaussian persistence curve given by Definition~\ref{def:GaussianPC} will be a Lipschitz function with respect to the input $t\in\mathbb{R}$ (see \cite{chung2020smooth}). Together with some mild assumptions on the weight functions $\kappa$, this implies that whenever there is is a process for randomly sampling persistence diagrams the associated Gaussian persistence curves will satisfy a version of the central limit theorem. See \cite{chung2020smooth} for details, or \cite{berry2018functional} for more general results about statistical properties of functional summaries.

\section{Stability}\label{sec:stability}
Given persistence diagrams $C$ and $D$, a \textbf{matching} between $C$ and $D$ is a bijection $\gamma: C\cup \Delta\to D\cup \Delta$ where $\Delta$ is the main diagonal in $\mathbb{R}^2$ with each point assigned infinite multiplicity. For a fixed matching $\gamma$ and $(b,d)\in C\cup\Delta$, we denote $\gamma(b,d)$ as $(\gamma_b,\gamma_d)$. We can compute the cost of a matching $\gamma$ as
\[L(\gamma) = \sum \|(b,d)-(\gamma_{b},\gamma_{d})\|_\infty\]
where the sum is over all points $(b, d)\in C\cup\Delta$ such that either $(b, d)\in C$ or $(\gamma_b, \gamma_d)\in D$. We define the \textbf{1-Wasserstein distance} $W_1(C,D)$ between diagrams $C$ and $D$ as the infimum of this cost function over all $\gamma$. That is \[W_1(C,D) =\inf_\gamma L(\gamma).\]

Our next goal is to show that if diagrams $C$ and $D$ are close with respect to the 1-Wasserstein distance, then the corresponding unweighted Gaussian persistence curves $G_C(t)$ and $G_D(t)$ are close with respect to the $L^1$ norm. In the general theory of summaries of persistence diagrams, this phenomenon is called {\bf stability}. 

 \begin{lemma}\label{lem:diff_gaussianCDF}
Let $d,d',\sigma\in \mathbb{R}$. Then \[\int_{-\infty}^\infty \left|\Phi\left(\frac{t-d'}{\sigma}\right)-\Phi\left(\frac{t-d}{\sigma}\right)\right|\dd{t} = |d-d'|.\]

\end{lemma}
\begin{proof}
Consider that \[\int_{-\infty}^\infty\left| \Phi\left(\frac{t-d'}{\sigma}\right)-\Phi\left(\frac{t-d}{\sigma}\right)\right|\dd{t} =\int_{-\infty}^\infty\left|\int_{\frac{t-d}{\sigma}}^{\frac{t-d'}{\sigma}} \phi(z)\dd{z}\right|\dd{t}\] Substitute $z = u + \frac{t-d}{\sigma}$. Then $\dd{z} = \dd{u}$ and
\begin{align}
    \int_{-\infty}^\infty\left|\int_{\frac{t-d}{\sigma}}^{\frac{t-d'}{\sigma}} \phi(z)\dd{z}\right|\dd{t} &= \int_{-\infty}^\infty \left|\int_{0}^{\frac{d-d'}{\sigma}} \phi\left(u + \frac{t-d}{\sigma}\right) \dd{u}\right| \dd{t}
    \\ &= \int_{-\infty}^\infty \int_{0}^{\frac{|d-d'|}{\sigma}}\phi\left(u + \frac{t-d}{\sigma}\right)~ \dd{t} \dd{u} \\ &=\int_{0}^{\frac{|d-d'|}{\sigma}} \int_{-\infty}^\infty \phi\left(u + \frac{t-d}{\sigma}\right)~ \dd{t} \dd{u}\\&= \int_{0}^{\frac{|d-d'|}{\sigma}} \left[ \sigma \Phi\left(u + \frac{t-d}{\sigma}\right)\right]_{-\infty}^{\infty} \dd{u} \\
    & = \int_{0}^{\frac{|d-d'|}{\sigma}} \sigma \dd{u} = |d-d'|.
\end{align}

\end{proof}

\begin{lemma}
\label{lem:g1-g2 bdd by W}
\begin{equation}
\int_{-\infty}^{\infty} \bigg|  \Phi(\frac{t-b_1}{\sigma})\Phi(\frac{d_1-t}{\sigma}) -\Phi(\frac{t-b_2}{\sigma})\Phi(\frac{d_2-t}{\sigma})\bigg|\dd{t} \leq |b_1 - b_2| + |d_1 - d_2|.
\end{equation}
\end{lemma}

\begin{proof}
\begin{align*}
    &\bigg|\int_{-\infty}^{\infty}\Phi(\frac{t-b_1}{\sigma})\Phi(\frac{d_1-t}{\sigma}) -\Phi(\frac{t-b_2}{\sigma})\Phi(\frac{d_2-t}{\sigma})\dd{t}\bigg|\\
    = &\bigg|\int_{-\infty}^{\infty} \Phi(\frac{t-b_1}{\sigma})\Phi(\frac{d_1-t}{\sigma}) -\Phi(\frac{t-b_2}{\sigma})\Phi(\frac{d_2-t}{\sigma}) +\Phi(\frac{t-b_1}{\sigma})\Phi(\frac{t-d_2}{\sigma}) -\Phi(\frac{t-b_1}{\sigma})\Phi(\frac{t-d_2}{\sigma})     \dd{t}\bigg|\\
    =&\bigg|\int_{-\infty}^{\infty}\Phi(\frac{t-b_1}{\sigma})[\Phi(\frac{t-d_2}{\sigma})-\Phi(\frac{t-d_1}{\sigma})]+[1-\Phi(\frac{t-d_2}{\sigma})][\Phi(\frac{t-b_1}{\sigma}-\Phi(\frac{t-b_2}{\sigma}))]\dd{t}\bigg|.
    \end{align*}
    Therefore, by Lemma~\ref{lem:diff_gaussianCDF} and $0\leq \Phi(t) \leq 1$,
    \begin{align*}
    &\leq 1\int_{-\infty}^{\infty} \bigg| \Phi(\frac{t-d_2}{\sigma})- \Phi(\frac{t-d_1}{\sigma})\bigg| \dd{t} +  1\int_{-\infty}^{\infty} \bigg| \Phi(\frac{t-b_1}{\sigma})- \Phi(\frac{t-b_2}{\sigma})\bigg| \dd{t} \\
&\leq | d_1 - d_2 | + | b_1 - b_2|.
\end{align*}
\end{proof}

Before proceeding to the first stability result, we will fix some notation that will be used in the remainder of this section. Recall that we defined $\delta_D=\min_{(b,d)\in D}(d-b)$. We will make the convention that when $D = \varnothing$, $\delta_D := \inf_{D}(d-b) := \infty$. For two diagrams $C$ and $D$, we further define $\delta_{C, D}=\min\{\delta_C, \delta_D, 1\}$. Now, fix a minimal cost matching $\gamma$ between two persistence diagrams $C$ and $D$. Let $C'$ be the points of $C$ which match with points of $D$ under $\gamma$, and let $D'$ be the image of $C'$ under this matching. Let $E=(C\setminus C')\cup (D\setminus D')$, that is $E$ consists of the points of $C$ and $D$ which match to the main diagonal under $\gamma$. We denote by $G_E(t)$ the unweighted Gaussian persistence curve on $E$. We note that since $\gamma$ has been assumed to be a minimal cost matching, all points of $E$ are matched to their closest point on the main diagonal under $\gamma$. In particular, if $(b, d)\in E$, then the contribution that this point makes to the cost of $\gamma$ is $\frac{d-b}{2}$. Finally, we enumerate the points of $C'$ as $\{(b_1, d_1),...,(b_N, d_N)\}$, and denote the image of an enumerated point under $\gamma$ as $\gamma(b_i,d_i) = (\gamma_{b_i},\gamma_{d_i})$.

We can now proceed to prove our first stability result for the case of unweighted Gaussian persistence curves. 
\begin{theorem}\label{t:unwstab}
Let $C$ and $D$ be persistence diagrams and let $G_C(t)$ and $G_D(t)$ be their unweighted Gaussian persistence curves. Then 
\begin{equation}
    \| G_C(t) - G_D(t) \|_1 \leq AW_1(C,D)
\end{equation}
Where $A=\max\left\{2,2\left(1+\frac{\sigma}{\sqrt{\pi}\delta_{C, D}}\right)\right\} $
\end{theorem}
\begin{proof}

We first consider the case when neither $C$ nor $D$ are the empty diagram.
\begin{align*}
    &\| G_C(t) - G_D(t) \|_1 \\ 
    & \leq \int_{-\infty}^{\infty}| \sum_{i=1}^N \Phi(\frac{t-b_i}{\sigma}) \Phi(\frac{d_i-t}{\sigma}) - \Phi(\frac{t-\gamma_{b_i}}{\sigma}) \Phi(\frac{\gamma_{d_i}-t}{\sigma})|+|\sum_{(b, d)\in E}\Phi(\frac{t-b_i}{\sigma}) \Phi(\frac{d_i-t}{\sigma})|~\dd{t}\\
    &\leq \sum_{i=1}^N|d_i-\gamma_{d_i}| + |b_i-\gamma_{b_i}|+\|G_E(t)\|_1\\
    &\leq 2\sum_{i=1}^N\|(b_i,d_i)-(\gamma_{b_i},\gamma_{d_i})\|_{\infty} + \left(1+\frac{\sigma}{\sqrt{\pi}\delta_E}\right)\sum_{(b,d)\in E}(d-b)\\
    &\leq \max\left\{2, 2\left(1+\frac{\sigma}{\sqrt{\pi}\delta_{C,D}} \right)  \right\}W_1(C,D)\\ 
    &= AW_1(C,D).
\end{align*}
Here the second inequality follows from Lemma \ref{lem:g1-g2 bdd by W} and Corollary \ref{unweighted1bound}.

\end{proof}

One might hope to remove terms depending on the diagram from $A$. However, some term involving $\delta_{C,D}$ will necessarily appear in this bound. To see this, suppose $D=C\cup E$ where $E$ consists of $k$ points with total lifespan small when compared to $\delta_C$. In other words suppose $L_E \leq \varepsilon$ and $\varepsilon$ is sufficiently small. Further, note that the optimal matching between diagrams $C$ and $D$ matches all points of $E$ to the diagonal. Then $W_1(C,D) = L_E \leq \varepsilon$. However, one can see by way of Proposition \ref{1-normGPC} that $\|G_C(t)-G_D(t)\|_1=\|G_E(t)\|_1\geq \frac{\sigma k}{\sqrt{\pi}}$. For this reason it seems natural to consider weights $\kappa(b,d)$ which go to zero as $b\to d$. E.g. one can take the lifespan function $\kappa(b,d) = d-b$.

We now consider the case where $G_{C, \kappa_C}(t)$ and $G_{D, \kappa_D}(t)$ are weighted Gaussian persistence curves on diagrams $C$ and $D$. We will derive a bound for arbitrary weights and then show how this bound can be improved by imposing restrictions on the weight functions. 

We keep the same notation as before, but in this case we also need to define the following notation for the weighting function. For each $(b, d)\in E$, define $\kappa_E(b, d)$ to be $\kappa_C(b, d)$ if $(b,d)\in C$ or to be $\kappa_D(b, d)$ if $(b, d)\in D$. Let $M_C=\max_{(b, d)\in C}|\kappa_C(b, d)|$ and define $M_D$, $M_E$, $M_{C'}$, and $M_{D'}$ similarly. Let $M_{C, D}=\max\{M_C, M_D\}$, and $M_{\gamma}=\max_{(b, d)\in C'}|\kappa_C(b, d)-\kappa_D(\gamma_b, \gamma_d)|$.

\begin{theorem}\label{t:wstab1}
\begin{equation}
    \| G_{C, \kappa_C}(t)-G_{D, \kappa_D}(t) \|_1 \leq BW_1(C, D)+M_\gamma\|G_{D'}\|_1
\end{equation}
with $B=\max\{2 M_C, 2M_E\left(1+\frac{\sigma} {\sqrt{\pi}\delta_E}\right)\}$.
\end{theorem}
\begin{proof}

In order to simplify notation, let $w_i=\kappa_C(b_i, d_i)$, $u_i=\kappa_D(\gamma_{b_i}, \gamma_{c_i})$, $F_i(t)=\Phi(\frac{t-b_i}{\sigma})\Phi(\frac{d_i-t}{\sigma})$, and $H_i(t)=\Phi(\frac{t-\gamma_{b_i}}{\sigma})\Phi(\frac{\gamma_{d_i}-t}{\sigma})$. With this notation, the CDF realizations of $G_{C, \kappa_C}(t)$ and $G_{D, \kappa_D}(t)$ are given by
\[
G_{C, \kappa_C}(t)=\sum_{i=1}^N w_iF_i(t)+\sum_{(b, d)\in C\setminus C'}\kappa_C(b, d)\Phi(\frac{t-b}{\sigma}) \Phi(\frac{d-t}{\sigma})
\]

\[
G_{D, \kappa_D}(t)=\sum_{i=1}^N u_iH_i(t)+\sum_{(b, d)\in D\setminus D'}\kappa_D(b, d)\Phi(\frac{t-b}{\sigma}) \Phi(\frac{d-t}{\sigma})
\]
Hence
\begin{align*}
    &\| G_{C, \kappa_C} - G_{D, \kappa_D} \|_1\\ 
    &\leq \sum_{i=1}^N\int_{-\infty}^{\infty}|w_iF_i-u_iH_i |~\dd{t}+\sum_{(b, d)\in E}\int_{-\infty}^{\infty}|\kappa_E(b, d)\Phi(\frac{t-b}{\sigma}) \Phi(\frac{d-t}{\sigma})|~\dd{t}\\
\end{align*}

We consider each component of this sum separately. By Corollary \ref{weighted1bound} we have,
\begin{align*}
&\sum_{(b, d)\in E}\int_{-\infty}^{\infty}|\kappa_E(b, d)\Phi(\frac{t-b}{\sigma}) \Phi(\frac{d-t}{\sigma})|~\dd{t}=\|G_{E, \kappa_E}\|_1\\
 &\leq \sum_{(b, d)\in E} \left(1+\frac{\sigma}{\sqrt{\pi}\delta_E}\right)\frac{2M_E(d-b)}{2}
\end{align*}
Now we consider $\sum_{i=1}^N\int_{-\infty}^{\infty}|w_iF_i-u_iH_i |~\dd{t}$.
\begin{align*}
    &\sum_{i=1}^N\int_{-\infty}^{\infty}|w_iF_i-u_iH_i |~\dd{t} \\
    &\leq \sum_{i=1}^N\int_{-\infty}^{\infty}|w_iF_i-w_iH_i+w_iH_i-u_iH_i |~\dd{t}\\
    &\leq \sum_{i=1}^Nw_i\int_{-\infty}^{\infty}|F_i-H_i|~\dd{t}+|w_i-u_i|\int_{-\infty}^\infty |H_i(t) |~\dd{t}\\
    &\leq 2M_C\sum_{i=1}^N\|(b_i,d_i)-(\gamma_{b_i},\gamma_{d_i})\|_\infty+M_\gamma\|G_{D'}\|_1
    \end{align*}
Combining the above inequalities and letting $B=\max\{2 M_C, 2M_E\left(1+\frac{\sigma}{\sqrt{\pi}\delta_E}\right)\}$,

\[
\| G_{C, \kappa_C} - G_{D, \kappa_D} \|_1\leq BW_1(C, D)+M_\gamma\|G_{D'}\|_1.
\]

\end{proof}

 One might hope to improve this bound by removing the additive term, $M_{\gamma}\|G_{D'}\|_1$. However, we can show that without additional assumptions placed on the weighting functions such an improvement is impossible. Indeed, suppose that $C=\{b, d\}$, $D=\{b+1, d+1\}$, $\kappa_C(b, d)=1$ and $\kappa_D(b+1, d+1)=u$. Then for all $d-b$ sufficiently large, $W_1(C, D)=1$ but  
 \[
\| G_{C, \kappa_C} - G_{D, \kappa_D} \|_1\geq (u-1)\|G_D\|_1
\]
 
 In other words, $\| G_{C, \kappa_C} - G_{D, \kappa_D} \|_1$ goes to infinity as either $u$ or $d-b$ goes to infinity.

To avoid this potentiality we will impose additional restrictions on the weighting functions $\kappa_C$ and $\kappa_D$. Precisely,  we will assume that there exists a constant $K$ such that for all $(x, y)$ and $(w, z)$,
\[
|\kappa_C(x, y)-\kappa_D(w, z)|\leq K\|(x, y)-(w, z)\|_\infty
.\]

This holds, for example, when $\kappa_C=\kappa_D$ is $K$--lipschitz with respect to the $L^\infty$ norm on $\mathbb R^2$. Using the notation of the above proof we will have
\[
\sum_{i=1}^N |w_i-u_i|\leq \sum_{i=1}^N  K\|(b_i,d_i)-(\gamma_{b_i},\gamma_{d_i})\|_\infty
\]

We can then replace $M_\gamma$ with $KW_1(C, D)$ in the above bound and thus obtain a bound that does not depend on the matching $\gamma$.

\begin{corollary}\label{c:wstab2}
Suppose there exists a constant $K$ such that for all $(x, y)$ and $(w, z)$,
\[
|\kappa_C(x, y)-\kappa_D(w, z)|\leq K\|(x, y)-(w, z)\|_\infty
\]
Let 
\begin{align*}
J&=\max\{B, K\|G_{D'}\|_1\}=\max\{2 M_C, 2M_E\left(1+\frac{\sigma} {\sqrt{\pi}\delta_E}\right), K\|G_{D'}\|_1\}\\
&\leq \max\{2 M_C, 2M_{C,D}\left(1+\frac{\sigma}{\sqrt{\pi}\delta_{C,D}}\right), K\|G_{D}\|_1\}.
\end{align*}
Then,
\[
  \| G_{C, \kappa}(t)-G_{D, \kappa}(t) \|_1 \leq JW_1(C, D).
\]
\end{corollary}

We now consider a few natural weight functions, starting with the lifespan function $\ell(b, d)=d-b$. Note that $\ell$ is $1$--lipschitz. We proceed as in the proof of the previous theorem, but in this case we split $E$ into disjoint subdiagrams $E'$ and $E''$ where $E'=\{(b, d)\in E\;|\; d-b\geq 1\}$ and $E''=E\setminus E'$.

\begin{align*}
&\sum_{(b, d)\in E} (d-b)(d-b+\frac{\sigma}{\sqrt{\pi}})\\
&=\sum_{(b, d)\in E'} (d-b)(d-b+\frac{\sigma}{\sqrt{\pi}})+\sum_{(b, d)\in E''} (d-b)(d-b+\frac{\sigma}{\sqrt{\pi}})\\
&\leq \sum_{(b, d)\in E'}\left( M_{E'}+\frac{\sigma}{\sqrt{\pi}}\right)(d-b)+\sum_{(b, d)\in E''} (d-b)(1+\frac{\sigma}{\sqrt{\pi}})\\
&\leq \max\{2\left(M_{E}+\frac{\sigma}{\sqrt{\pi}}\right), 2+\frac{2\sigma}{\sqrt{\pi}}\}\sum_{(b, d)\in E}(d-b)
\end{align*}

Since each $(b, d)\in E$ contributes $\frac{d-b}{2}$ to the cost of $\gamma$, we obtain
\begin{theorem}\label{t:wstab3}
 Let $G=\max\{2 M_C, \|G_{D'}\|_1,2\left(M_{E}+\frac{\sigma}{\sqrt{\pi}}\right), 2+\frac{2\sigma}{\sqrt{\pi}}\}\leq \max\{2 M_C, \|G_{D}\|_1, 2\left(M_{C,D}+\frac{\sigma}{\sqrt{\pi}}\right), 2+\frac{2\sigma}{\sqrt{\pi}}\}$. Then

\[
  \| G_{C, \ell}(t)-G_{D, \ell}(t) \|_1 \leq G W_1(C, D).
\]
\end{theorem}

We can remove the dependence of our bound on $\|G_D\|_1$ by normalizing the lifespan functions. Recall that $L_D$ is the total lifespan of the diagram $D$.

Recall that by Corollary \ref{unweighted1bound} 
\[
\|G_D\|_1\leq \left(1+\frac{\sigma}{\sqrt{\pi}\delta_D}\right)L_D.
\]
It is also straightforward to check (see \cite[Section 5.2]{chunglawson2019}) that
\[
|L_C-L_D|\leq 2W_1(C, D).
\]

Define $\hat{\ell}_D(b, d)=\frac{d-b}{\overline{L}_D}$. We will assume that diagrams $C$ and $D$ both have total lifespan at least $1$, and hence $M_{C, D}\leq 1$.
\begin{theorem}\label{t:wstab4}
\[
  \| G_{C, \hat{\ell}_C}(t)-G_{D, \hat{\ell}_D}(t) \|_1 \leq P W_1(C, D).
\]
Where $P=\max\{2, 2\left(M_E+\frac{\sigma}{\sqrt{\pi}}\right), 2+\frac{2\sigma}{\sqrt{\pi}}, 4+\frac{4\sigma}{\sqrt{\pi}\delta_D}\}\leq \max\{2, 2+\frac{2\sigma}{\sqrt{\pi}}, 4+\frac{4\sigma}{\sqrt{\pi}\delta_D}\}$
\end{theorem}
\begin{proof}

The $2M_C=2$ term and the $\max\{\frac{M_E\sigma}{\sqrt{\pi}}, 1+\frac{\sigma}{\sqrt{\pi}}\}$ terms can be computed as in the general weights theorem and the lifespan computation above (noting that $\frac{d-b}{L_D}$ is always less then $d-b$)

We focus on the $M_\gamma\|G_{D'}\|$ term in the theorem with general weights. this term comes from bounding 
\[
\sum_{i=1}^N|w_i-u_i|\int_{-\infty}^\infty |H_i(t) |~\dd{t}
\]

\begin{align*}
   &\sum_{i=1}^N|\frac{d_i-b_i}{L_C}-\frac{\gamma_{d_i}-\gamma_{b_i}}{L_D}|\int_{-\infty}^\infty |H_i(t) |~\dd{t} \\
   &\leq \sum_{i=1}^N|\frac{d_i-b_i}{L_C}-\frac{d_i-b_i}{L_D}+\frac{d_i-b_i}{L_D}-\frac{\gamma_{d_i}-\gamma_{b_i}}{L_D}|\int_{-\infty}^\infty |H_i(t) |~\dd{t}\\
   &\leq \sum_{i=1}^N|\frac{d_i-b_i}{L_C}\frac{L_D-L_C}{L_D}|+|\frac{(d_i-b_i)-(\gamma_{d_i}-\gamma_{b_i})}{L_D}|\int_{-\infty}^\infty |H_i(t) |~\dd{t}\\
   &\leq \sum_{i=1}^N\frac{4}{L_D}W_1(C, D)\int_{-\infty}^\infty |H_i(t) |~\dd{t}\\
   &\leq\frac{4}{L_D}W_1(C, D)\|G_D\|_1\\
  &\leq\left(4+\frac{4\sigma}{\sqrt{\pi}\delta_D}\right)W_1(C, D).
\end{align*}

\end{proof}

\section{Injectivity}
\label{sec:inj}
In this section we study the injectivity of the transformation from a persistence diagram to either the persistence surface or the corresponding Gaussian persistence curve. By injectivity here we mean that distinct diagrams produce distinct persistence surfaces or curves. In general, this notion depends on the choice of weight functions, see example \ref{ex:counter ex gpc is invertible}. However, we show the injectivity of unweighted persistence surfaces and (in most cases) unweighted Gaussian persistence curves. We also conjecture that any weight functions defined independently of the diagrams will produce injective persistence surfaces and curves.

We first show injectivity for unweighted persistence surfaces. Given two diagrams $C$ and $D$, this amounts to showing that the multi-set of means in the corresponding surfaces is equal. We achieve this by setting up an infinite system of equations, which can only be solved when the multi-sets of means are exactly equal. We will first give some technical lemmas.

\begin{lemma}\label{lem:equal-projections}
Let $A,B\subset\mathbb{R}$ be finite sets of equal cardinality. Suppose \[\sum_{a\in A}\phi\left(\frac{x-a}{\sigma}\right) = \sum_{b\in B}\phi\left(\frac{x-b}{\sigma}\right).\]
Then for every $n\in\mathbb{N}$, \[\sum_{a\in A}a^n = \sum_{b\in B}b^n\]
\end{lemma}

\begin{proof}
Let $n\in \mathbb{N}$ and suppose

\[\sum_{a\in A}\phi\left(\frac{x-a}{\sigma}\right) = \sum_{b\in B}\phi\left(\frac{x-b}{\sigma}\right).\]

Multiplying by $\frac{x^n}{\sigma}$ and integrating over $\mathbb{R}$ with respect to $x$ gives 
\begin{align}
    \sum_{a\in A} \int_{\mathbb{R}}x^n\frac{1}{\sigma}\phi(\frac{x-a}{\sigma})dx &= \sum_{b\in B} \int_{\mathbb{R}}x^n\frac{1}{\sigma}\phi(\frac{x-b}{\sigma})dx\\
    \label{moments}\sum_{a\in A} \sum_{k \text{ even}}^n {n \choose k}a^{n-k}\sigma^k(k-1)!! &= \sum_{b\in B} \sum_{k \text{ even}}^n {n \choose k}b^{n-k}\sigma^k(k-1)!!.
\end{align}
Here (\ref{moments}) follows from computing the $n^{th}$ moment of the normal distribution. We will now prove that $\sum_{a\in A}a^n = \sum_{b\in B}b^n$ for all $n \in \mathbb{N}$ by induction. When $n=1$ this follows immediately from equation (\ref{moments}) above. Now assume that $\sum_{a\in A}a^m = \sum_{b\in B}b^m$ for all $m<n$. Then again by equation (\ref{moments}) we have, 
\begin{align*}
       \sum_{a\in A} \sum_{k \text{ even}}^n {n \choose k}a^{n-k}\sigma^k(k-1)!! &= \sum_{b\in B} \sum_{k \text{ even}}^n {n \choose k}b^{n-k}\sigma^k(k-1)!!\\
       \sum_{k \text{ even}}^n{n \choose k}\sigma^k(k-1)!!\sum_{a\in A}a^{n-k} &= \sum_{k \text{ even}}^n{n \choose k}\sigma^k(k-1)!!\sum_{b\in B}b^{n-k}\\
       \sum_{a\in A}a^n &= \sum_{b\in B}b^n, 
\end{align*}
where the last step follows by our inductive hypothesis.
\end{proof}

\begin{lemma}\label{lem:injectiveproduct}
Let $A,B\subset \mathbb{R}^2$ be finite sets of equal cardinality. Suppose \[\sum_{(a,b)\in A}\phi\left(\frac{x-a}{\sigma}\right)\phi\left(\frac{y-b}{\sigma}\right) = \sum_{(\alpha,\beta)\in B}\phi\left(\frac{x-\alpha}{\sigma}\right)\phi\left(\frac{y-\beta}{\sigma}\right).\]
Then for any $m_1,m_2\in \mathbb{N}$, \[\sum_{(a,b)\in A}a^{m_1}b^{m_2} = \sum_{(\alpha,\beta)\in B}\alpha^{m_1}\beta^{m_2}. \]
\end{lemma}
\begin{proof}
Let $m_1,m_2\in \mathbb{N}$ and suppose \[\sum_{(a,b)\in A}\phi\left(\frac{x-a}{\sigma}\right)\phi\left(\frac{y-b}{\sigma}\right) = \sum_{(\alpha,\beta)\in B}\phi\left(\frac{x-\alpha}{\sigma}\right)\phi\left(\frac{y-\beta}{\sigma}\right).\]

Multiplying by $x^{m_1}y^{m_2}$ and integrating over $\mathbb{R}$ with respect to $x$ and $y$ yields,  
\begin{align}
        \nonumber&\sum_{(a,b)\in A} \int_{\mathbb{R}}x^{m_1}\phi\left(\frac{x-a}{\sigma}\right)dx\int_{\mathbb{R}}y^{m_2}\phi\left(\frac{y-b}{\sigma}\right)dy \\
        = &\sum_{(\alpha,\beta)\in B} \int_{\mathbb{R}}x^{m_1} \phi\left(\frac{x-\alpha}{\sigma}\right)dx\int_{\mathbb{R}}y^{m_2}\phi\left(\frac{y-\beta}{\sigma}\right)dy\\\nonumber
    &\sum_{(a,b)\in A} \sum_{k \text{ even}}^{m_1} {m_1 \choose k}a^{m_1-k}\sigma^k(k-1)!!\sum_{l\text{ even}}^{m_2} {m_2 \choose l}b^{m_2-l}\sigma^l(l-1)!! \\
    = &\sum_{(\alpha,\beta)\in B} \sum_{k \text{ even}}^{m_1} {m_1 \choose k}\alpha^{m_1-k}\sigma^k(k-1)!!\sum_{l\text{ even}}^{m_2} {m_2 \choose l}\beta^{m_2-l}\sigma^l(l-1)!!.\label{equ:equalmoments}
\end{align}
Now when $m_1=m_2=1$, equation \ref{equ:equalmoments} yields \[\sum_{(a,b)\in A}ab = \sum_{(\alpha,\beta)\in B}\alpha\beta.\] We proceed by proving that the claim is true for all $m_1+m_2 =n \in \mathbb{N}$ by strong induction on $n$. The base case has been proven above so assume that the statement is true for all $m_1+m_2=j<n\in\mathbb{N}$. Note that $m_1-k+m_2-l =j$ where $j \in \mathbb{N}$ and $j<n$ for all $k,l\in 2\mathbb{N}$ with $k\leq m_1,l\leq m_2$. Then the inductive hypothesis and equation \ref{equ:equalmoments} yields $$\sum_{(a,b)\in A}a^{m_1}b^{m_2} = \sum_{(\alpha,\beta)\in B}\alpha^{m_1}\beta^{m_2}$$ as desired.
\end{proof}
\begin{theorem}\label{injective}
Let $C$ and $D$ be persistence diagrams and let $\rho_C$ and $\rho_D$ be the corresponding unweighted persistence surfaces. If $\rho_C\equiv\rho_D$, then $C=D$.
\end{theorem}

\begin{proof}
We will denote the points of $C$ by $(b_i^C,d_i^C)$ and analogously the points of $D$ will be denoted $(b_i^D,d_i^D)$. Assume that $\rho_C=\rho_D$. We first note that $\int_{\mathbb{R}^2}\rho_C = |C|$. Thus, we must have $|C|=|D|$. We will assign $N = |C|$. Thus, we have \[\sum_{i=1}^N\phi\left(\frac{x-b_i^C}{\sigma}\right)\phi\left(\frac{y-d_i^C}{\sigma}\right) = \sum_{i=1}^N\phi\left(\frac{x-b_i^D}{\sigma}\right)\phi\left(\frac{y-d_i^D}{\sigma}\right).\]

Integrating over $\mathbb{R}$ with respect to $y$ and dividing by $\sigma$ yields \[\sum_{i=1}^N\phi\left(\frac{x-b_i^C}{\sigma}\right) = \sum_{i=1}^N\phi\left(\frac{x-b_i^D}{\sigma}\right).\]
 An application of Lemma~\ref{lem:equal-projections} yields $\sum_{i=1}^N(b_i^C)^n = \sum_{i=1}^N(b_i^D)^n$ for all $n\in\mathbb{N}$. A similar method yields $\sum_{i=1}^N(d_i^C)^n = \sum_{i=1}^N(d_i^D)^n$. Next, an application of Lemma~\ref{lem:injectiveproduct} yields that for each $m_1,m_2\in\mathbb{N}$ we have $\sum_{i=1}^N(b_i^C)^{m_1}(d_i^C)^{m_2} = \sum_{i=1}^N(b_i^D)^{m_1}(d_i^D)^{m_2}$.

Before proceeding we recall the following basic fact:
suppose $a_1,..., a_k$ are non-negative real numbers and $c_1,...,c_k$ are positive real numbers such that $a_k>a_i$ for all $1\leq i<k$. Then there exists $m\in \mathbb{N}$ such that
\[
c_ka_k^m>\sum_{i=1}^{k-1} c_ia_i^m.
\]

Now, label the points in $C$ such that $b_1^C\leq b_2^C...\leq b_N^C$ and whenever $b_i^C=b_{i+1}^C$, $d_i^C\leq d_{i+1}^C$. Label the points in $D$ in the same way.

Suppose that some $b_k^C\neq b_k^D$. Choosing $k$ to be the largest such index, it follows from above that that for all $n\geq 1$.
 \[
\sum_{i=1}^k (b_i^C)^n=\sum_{i=1}^k (b_i^D)^n.
\]
 
 Without loss of generality, we assume that $b_k^C>b_k^D$. Since $b_k^D\geq b_{k-1}^D\geq...\geq b_1^D$, we can find an $m$ such that 
 \[
 \sum_{i=1}^k (b_i^C)^m\geq (b_k^C)^m>\sum_{i=1}^k (b_i^D)^m
 \]
 which is a contradiction. Hence we have that $b_i^C=b_i^D$ for all $1\leq i\leq N$.
 
 Now suppose that for some $k$, $d_k^C\neq d_k^D$. Again we choose $k$ to be largest such index and assume without loss of generality that $d_k^C>d_k^D$. Now we choose $m_1$ such that $(b_k^C)^{m_1}d_k^C>(b_i^D)^{m_1}d_i^D$ for all $1\leq i< k$. $m_1$ exists since for all $1\leq i\leq k$, either $b_k^C=b_k^D>b_i^D$, or $b_k^C=b_k^D=b_i^D$ and $d_k^C>d_k^D\geq d_i^D$.
 
 Given this $m_1$, we can apply the above fact again to find $m_2$ such that
 \[
 \sum_{i=1}^k((b_i^C)^{m_1}d_i^C)^{m_2}\geq ((b_k^C)^{m_1}d_k^C)^{m_2}> \sum_{i=1}^k((b_i^D)^{m_1}d_i^D)^{m_2}
 \]
 which is another contradiction. Hence, we must have that $d_i^C=d_i^D$ for all $1\leq i\leq N$ which means we have shown that $C=D$.

\end{proof}

We prove a partial result for the unweighted Gaussian Persistence Curve.

\begin{theorem}\label{t:ainj}
Let $C$ and $D$ be two persistence diagrams with maximum death values $d^C_{max},d_{max}^D$ and minimum birth values $b_{min}^C, b_{min}^D$ respectively. Suppose that either $d_{max}^C \neq d_{max}^D$ or $b_{min}^C \neq b_{min}^D$. Then $G_C\neq G_D$.
\end{theorem}

\begin{proof}
Let $(b_1, d_1)$ be a point of $C$ and $(b_2, d_2)$ a point of $D$. 
We first show that if $d_1 > d_2$ then, 
\begin{equation}\label{cdf}
    \lim_{t \to \infty} \frac{\Phi(\frac{t-b_2}{\sigma})\Phi(\frac{d_2-t}{\sigma})}{\Phi(\frac{t-b_1}{\sigma})\Phi(\frac{d_1-t}{\sigma})} = 0.
\end{equation}
We will proceed by cases. If $b_1 = b_2$ then with L'hospital's rule we have, 

\begin{equation}
    \lim_{t \to \infty} \frac{\Phi(\frac{t-b_2}{\sigma})\Phi(\frac{d_2-t}{\sigma})}{\Phi(\frac{t-b_1}{\sigma})\Phi(\frac{d_1-t}{\sigma})} = \lim_{t \to \infty} \frac{\phi(\frac{d_2-t}{\sigma})}{\phi(\frac{d_1-t}{\sigma})} =
    \lim_{t \to \infty} e^{(\frac{d_2+d_1-2t}{\sigma})(\frac{d_1-d_2}{\sigma})} = 0.
\end{equation}
Now if $b_1 < b_2$ we note that $\Phi(\frac{t-b_1}{\sigma}) > \Phi(\frac{t-b_2}{\sigma})$ for any sufficiently large value of $t$. Applying this inequality to the limit we reduce back to the first case. Finally if $b_1 > b_2$ we have, 

\begin{equation}
    \lim_{t \to \infty} \frac{\Phi(\frac{t-b_2}{\sigma})\Phi(\frac{d_2-t}{\sigma})}{\Phi(\frac{t-b_1}{\sigma})\Phi(\frac{d_1-t}{\sigma})} = \lim_{t \to \infty} \frac{\Phi(\frac{t-b_2}{\sigma})}{\Phi(\frac{t-b_1}{\sigma})}\lim_{t \to \infty}\frac{\Phi(\frac{d_2-t}{\sigma})}{\Phi(\frac{d_1-t}{\sigma})} = 0.
\end{equation}

Assume that $d^C_{max}>d^D_{max}$, and hence $d^C_{max}$ is larger then the death value of any point of $D$. Let $(b^C, d^C_{max})$ be a point of $C$. Then
\begin{align*}
\lim_{t \to \infty} \frac{G_D(t)}{G_C(t)}&= \lim_{t \to \infty} \frac{\displaystyle\sum_{(b, d)\in D}\Phi(\frac{t-b}{\sigma})\Phi(\frac{d-t}{\sigma})}{\displaystyle\sum_{(b, d)\in C}\Phi(\frac{t-b}{\sigma})\Phi(\frac{d-t}{\sigma})}\leq \lim_{t \to \infty} \frac{\displaystyle\sum_{(b, d)\in D}\Phi(\frac{t-b}{\sigma})\Phi(\frac{d-t}{\sigma})}{\Phi(\frac{t-b^C}{\sigma})\Phi(\frac{d^C_{max}-t}{\sigma})}\\
&=\displaystyle\sum_{(b, d)\in D}\lim_{t \to \infty}\frac{\Phi(\frac{t-b}{\sigma})\Phi(\frac{d-t}{\sigma})}{\Phi(\frac{t-b^C}{\sigma})\Phi(\frac{d^C_{max}-t}{\sigma})}=0
\end{align*}

Hence, for any sufficiently large value of $t$, $G_C(t)>G_D(t)$. A similar argument applied when assuming distinct minimum birth values, one simply needs to look at the limit as $t$ approaches negative infinity instead.

\end{proof}

Obtaining a result for general surfaces or curves may prove to be challenging. The next example shows that the injectivity of persistence surfaces (hence Gaussian persistence curves) cannot be generalized to arbitrary weight functions.
\begin{ex}
\label{ex:counter ex gpc is invertible}
Let $\kappa_D(b,d) = \frac{d-b}{L_D}$. Take $C = \{(b,d)\}$ and $D = \{(b,d),(b,d)\}$. That is, $C$ is a diagram with a single point and $D$ is a diagram with two points both at the same place as $C$. Then $D\neq C$ but 
\begin{align*}
\rho_{C, \kappa_C}(x, y)&=\phi\left(\frac{x-b}{\sigma}\right)\phi\left(\frac{y-d}{\sigma}\right)\\&=\frac12\phi\left(\frac{x-b}{\sigma}\right)\phi\left(\frac{y-d}{\sigma}\right)+\frac12\phi\left(\frac{x-b}{\sigma}\right)\phi\left(\frac{y-d}{\sigma}\right)=\rho_{D, \kappa_D}(x, y).
\end{align*}

\end{ex}

Note that the weight function in the above example is a natural one to consider in practice as it produced a strong stability result in the original persistence curve setting and also performed well in computer experiments \cite{chunglawson2019}. However, in this example at least the failure of injectivity of the persistence surface is clearly tied to the fact that different diagrams produce different weighting functions. We conjecture that this is the only way for injectivity to fail.
\begin{conjecture}
Let $\kappa\colon \mathbb R^2\to\mathbb R^+$ be a weighting function and let $C$ and $D$ be persistence diagrams. Suppose that $\rho_{C, \kappa}=\rho_{D, \kappa}$. Then $C=D$.
\end{conjecture}

Since $\rho_{C, \kappa_C}=\rho_{D, \kappa_D}$ implies that $G_{C, \kappa_C}=G_{D, \kappa_D}$, example \ref{ex:counter ex gpc is invertible} also shows that it is possible for distinct diagrams to produce the same Gaussian persistence curve. However, as with persistence surfaces we conjecture that this cannot happen for weighting functions which are independent of the diagrams.

\begin{conjecture}
Let $\kappa\colon \mathbb R^2\to\mathbb R^+$ be a weighting function and let $C$ and $D$ be persistence diagrams. Suppose that $G_{C, \kappa}=G_{D, \kappa}$. Then $C=D$.
\end{conjecture}

\bibliographystyle{plain}

\bibliography{Main.bib}

\end{document}